\documentclass[runningheads]{llncs}

\newif\ifArxiv
\Arxivtrue

\usepackage{amssymb,amsfonts,amsmath}
\usepackage{booktabs,tabularx,cite}
\usepackage{url,xspace,soul,wrapfig}
\usepackage{graphicx,paralist,enumerate}
\usepackage[hidelinks]{hyperref}
\usepackage{subfigure}
\usepackage[textsize=footnotesize]{todonotes}
\usepackage[basic]{complexity}
\usepackage[capitalise]{cleveref}

\usepackage[utf8]{inputenc}
\usepackage{amsmath}
\usepackage{mathtools}
\usepackage{xcolor}
\usepackage{url}
\ifArxiv
\else
\usepackage{nopageno}
\fi

\makeatletter
\AtBeginDocument{
	\def\doi#1{\url{https://doi.org/#1}}}
\makeatother

\graphicspath{{./Figures/}}

\renewcommand{\orcidID}[1]{\href{https://orcid.org/#1}{\includegraphics[scale=.03]{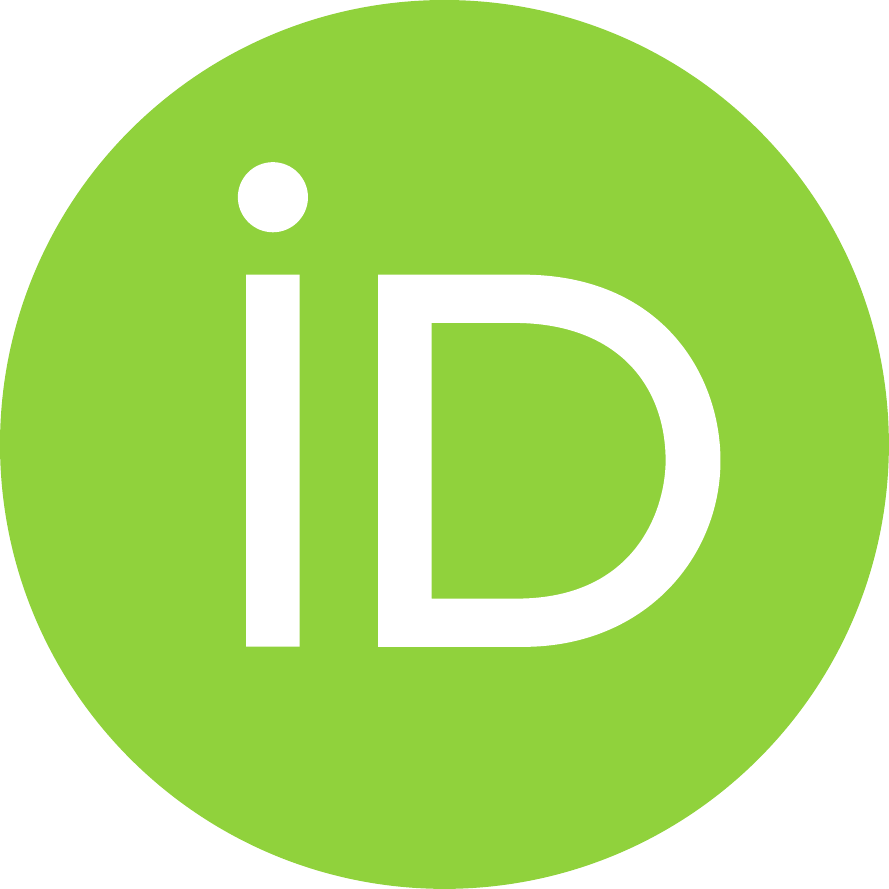}}} 
\makeatletter
\newenvironment{usecounterof}[2]{
	\def\@tempb{#1}
	\expandafter\renewcommand\csname thetheorem\endcsname{\ref{#2}}\@nameuse\@tempb}{
	\@nameuse{end\@tempb}\addtocounter{theorem}{-1}}
\makeatother

\newcommand{\mypar}[1]{\smallskip\noindent{\bfseries #1.}}

\begin{document}

\title{Layered Area-Proportional\\ Rectangle Contact Representations\thanks{We acknowledge funding by the Austrian Science Fund (FWF) under grant P31119.}}

\author{Martin Nöllenburg\inst{1}\orcidID{0000-0003-0454-3937} \and
Anaïs Villedieu\inst{1}\orcidID{0000-0001-6196-8347} \and
Jules Wulms\inst{1}}

\authorrunning{M. Nöllenburg, A. Villedieu, J. Wulms}

\institute{Algorithms and Complexity Group, TU Wien, Vienna, Austria\\
\email{\{noellenburg,avilledieu,jwulms\}@ac.tuwien.ac.at}}

\maketitle

\begin{abstract}

We investigate two optimization problems on area-pro\-por\-tio\-nal rectangle contact representations for layered, embedded planar graphs. The vertices are represented as interior-disjoint unit-height rectangles of prescribed widths, grouped in one row per layer, and each edge is ideally realized as a rectangle contact of positive length. 
Such rectangle contact representations find applications in semantic word or tag cloud visualizations, where a collection of words is displayed such that pairs of semantically related words are close to each other. 
In this paper, we want to maximize the number of realized rectangle contacts or minimize the overall area of the rectangle contact representation, while avoiding any false adjacencies.
We present a network flow model for area minimization, a linear-time algorithm for contact maximization of two-layer graphs, and an ILP model for maximizing contacts of $k$-layer graphs.

\end{abstract}

\section{Introduction}
	Contact representations of planar graphs are a well-studied topic in graph theory, graph drawing, and computational geometry~\cite{f-rsrpg-13,k-kka-36,fmr-tcg-94}. 
Vertices are represented by geometric objects, e.g., disks or polygons, and two objects touch if and only if they are connected by an edge.
They find many applications, for instance in VLSI design~\cite{ys-fgd2rm-93}, cartograms~\cite{nk-sc-16}, or semantic word clouds~\cite{wu_semantic-preserving_2011,barth_semantic_2014}.

\begin{figure}[h]
	\centering
	\includegraphics[width=1\linewidth]{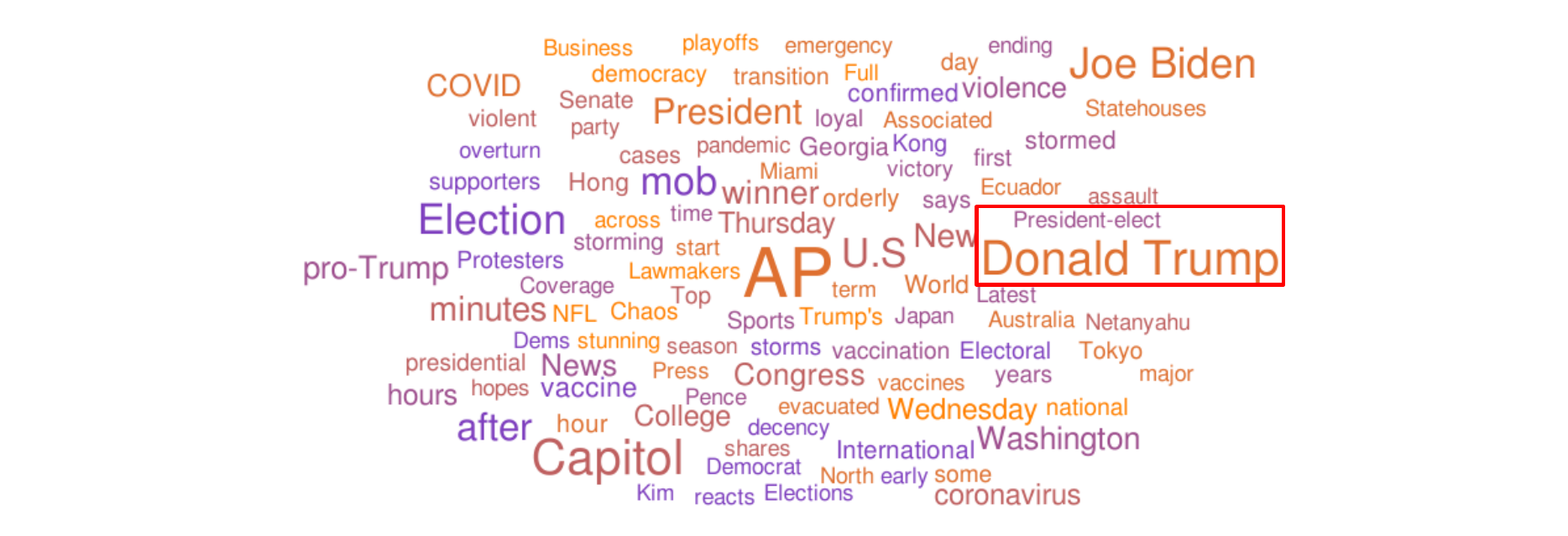}
	\caption{Word cloud generated from the \href{https://apnews.com}{apnews.com} frontpage by \href{https://worditout.com}{worditout.com} on the day of the certification of  Joe Biden's victory in the 2020 US elections.}
	\label{fig:FAex}
\end{figure}

Word or tag clouds are popular visualizations that summarize textual information in an aesthetically pleasing way. 
They show the main themes of a text by displaying the most important keywords obtained from text analysis and scale the word size to their frequency in the text. 
Word clouds became widespread after the first automated generation tool ``Wordle'' was published in 2009~\cite{viegas_participatory_2009}. 

Word clouds with their different font sizes and words packed without semantic context, such as the one shown in Fig.~\ref{fig:FAex}, have also received some criticism as their audience sometimes fails at understanding the underlying data (while enjoying their playful nature)~\cite{hearst_evaluation_2020}. 
For example, neighboring words that are not  semantically related can be misleading (see marked words in Fig.~\ref{fig:FAex}).
As a way to improve readability, \emph{semantic} word clouds have been introduced \cite{cui_context_2010,wu_semantic-preserving_2011,barth_semantic_2014}. 
In semantic word clouds, an underlying edge-weighted graph indicates the semantic relatedness of two words, whose positions are chosen 
such that semantically related words are next to each other while unrelated words are kept far apart.

Classic word clouds are often generated using forced-based approaches, alongside with a spiral placement heuristic \cite{viegas_participatory_2009, wang_edwordle:_2018, wang_shapewordle:_2020} that allows for a very compact final layout. This method is powerful even when the rough 
position of a word is dictated by an underlying map \cite{buchin_geo_2016, li_metro-wordle:_2018}.
Semantic word clouds on the other hand have been approached with many different techniques, e.g., force directed~\cite{cui_context_2010}, seam-carving~\cite{wu_semantic-preserving_2011}, and multidimensional scaling~\cite{barth_experimental_2014}. 
The problem has also been studied from a theoretical point of view, where an edge of the semantic word graph is realized if the bounding boxes of two related words properly touch; the realized edge weight is gained as profit.
Then the semantic word cloud problem can be phrased as the optimization problem to maximize the total profit. 
Barth et al.~\cite{barth_semantic_2014} and later Bekos et al.~\cite{bekos_improved_2017} gave several hardness and approximation results for this problem (and some variations) on certain graph classes. 
The underlying geometric problem also has links to more general contact graph representation problems, like rectangular layouts~\cite{buchsbaum_rectangular_2008} or cartograms~\cite{van_kreveld_rectangular_2007}.

In most of the literature about layered graphs, vertices are assigned to rows without a predefined left-to-right order, yet this has interesting properties in the context of word clouds. 
For instance, layered rectangle contact representations are compact, assuming a good assignment they have an even distribution of words and our eye naturally understands words grouped in rows or tables. 
In this paper we study row-based contact graphs of unit-height but arbitrary-width rectangles, which may represent the bounding boxes of words with fixed font size.

\mypar{Problem description} As input we take a layered graph~$G=(V,E)$ on $L$~layers, with an arbitrary number of vertices per layer. Each vertex $v_{i,j}\in V$ is indexed by its layer~$i\in[0,L-1]$ and its position~$j$ within the layer: $v_{i,j}$ is the $j^{th}$ vertex on the $i^{th}$ layer. The edge set~$E$ consists of edges connecting each vertex $v_{i,j}$ to its neighbors $v_{i,j-1}$ and $v_{i,j+1}$ on the same row (if they exist), and connections between adjacent rows form an internally triangulated graph.
We associate each vertex with an axis-aligned unit-height rectangle
$R_{i,j}$ with width $w_{i,j}$, and $y$-coordinate $i$. We want to compute its $x$-position $x_{i,j}$ given by the $x$-coordinate of its bottom left corner such that the rectangles do not overlap except on their boundaries (see Fig.~\ref{fig:grid}). Leaving whitespace between two rectangles on the same layer is allowed and forms a \emph{gap}.
Such a layout $\mathcal{R}$ is called a \emph{representation} of $G$. An edge $(u,v) \in E$ is \emph{realized} in a representation $\mathcal{R}$ if rectangles $R_u$ and $R_v$, representing vertices $u$ and $v$, intersect along their boundaries for a positive length $\varepsilon > 0$, which we denote by $(R_u,R_v) \in \mathcal{R}$. If $R_u$ and $R_v$ are horizontally adjacent we call the contact a \emph{horizontal contact}.

Otherwise, the intersection is located along a horizontal boundary if $R_u$ and $R_v$ are on adjacent layers; these are called \emph{vertical contacts}. 
Contacts between rectangles whose vertices are not adjacent in $G$ are \emph{false adjacencies}.
Such adjacencies can mislead a user to infer a link between unrelated words, invalidating the representation. Within this model we study two problem variations, \emph{area minimization} and \emph{contact maximization}.

\begin{figure}[tb]
	\centering
	\includegraphics[width=1\linewidth]{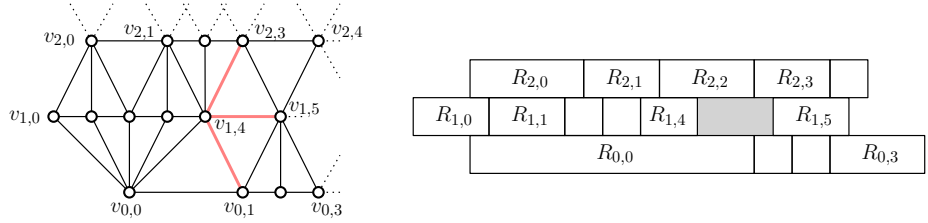}
	\caption{Partial drawing of a graph $G$, along with a representation $\mathcal{R}$ of the visible vertices of $G$. Red edges are not realized, due to the gray gap in $\mathcal{R}$.}
	\label{fig:grid}
\end{figure}

For the area minimization problem the goal is to produce a representation~$\mathcal{R}$ that minimizes the total width of the gaps in~$\mathcal{R}$. The contact maximization problem asks to maximize the number of adjacencies realized in~$\mathcal{R}$, as specified by edge set~$E$. For both optimization criteria, false adjacencies are forbidden: otherwise a trivial gap-less representation would always be a solution to the area minimization problem and in the case of contact maximization, false adjacencies may reduce the number of lost contacts with respect to a valid optimal solution as Fig.~\ref{fig:forbidden} shows. We say a representation is \emph{valid} if it has no false adjacencies.

\begin{figure}
	\centering
	\includegraphics{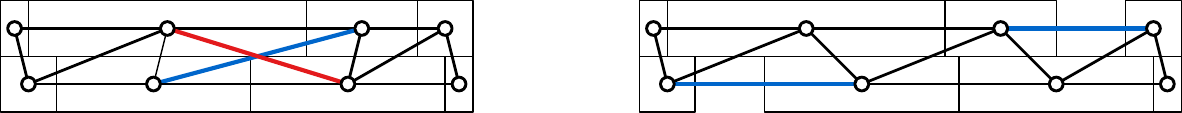}
	\caption{Allowing false adjacencies (red) could reduce lost adjacencies (blue).} 
	\label{fig:forbidden}
\end{figure}

\section{Area minimization}
To solve the area minimization problem, we construct a flow network $N=(G' = (V',E');l;c;b;cost)$ for a given vertex-weighted layered graph $G=(V,E)$, with edge capacity lower bound $l \colon E'\rightarrow \mathbb{R}^+_0$, edge capacity $c \colon E'\rightarrow \mathbb{R}^+_0$, vertex production/consumption $b \colon V'\rightarrow \mathbb{R}$ and cost function $cost \colon E'\rightarrow \mathbb{R}^+_0$. 
Each unit of cost will represent a unit length gap and each unit of flow on an edge will represent a unit length contact. To build the network we create two vertices $v^a$ and $v^b$ for each rectangle, that respectively receive the flow from the lower layer and output flow to the upper layer, and one for each potential gap, located between each sequential pair of rectangles in the same layer. Every edge $e$ that ends on a gap vertex has $cost(e) = 1$. We also add an edge $e$ between $v^a$ and $v^b$ for each $R_{i,j}$ with $l(e)=c(e)=w_{i,j}$ and no cost to ensure that rectangle nodes receive exactly as much flow as they are wide.

The intuition behind the network is that it represents a stack of layers consisting of rectangles and gaps, with a maximum width of $w_{\max}\cdot K$, $K$ being the maximum number of rectangles per layer, and $w_{\max}$ the width of the widest rectangle. To facilitate this flow on all layers, there are buffer vertices on both sides of each layer.
Each rectangle is as wide as the amount of flow its vertices $v^a$ and $v^b$ receive, and has contacts with its upper and lower neighbors as wide as the flow on the edges representing these contacts. 
Every vertex has edges to the layer above as far as its rectangle is allowed to have contacts: a rectangle $R_{i,j}$ that has only one upper neighbor, will have an edge to that neighbor, and to the gaps on that neighbor's right and left side. Any further edge would be to another rectangle with which $R_{i,j}$ should not share a contact, and such edges would hence result in false adjacencies. We picture the stack bottom-up, meaning that the flow comes in at the bottom layer and exits from the top layer. A gap block $g_{i,j}$ will reach as far left and right as its left and right neighbors in the same row: if rectangle $R_{i,j}$ lies directly left of $g_{i,j}$, then the furthest left upward neighbor of $R_{i,j}$ is the furthest left upward neighbor of $g_{i,j}$. If $g_{i,j}$ could reach even further, then it would essentially push $R_{i,j}$ into a false adjacency. \ifArxiv The exact construction is detailed in Appendix~\ref{app:flow} and sketched in Fig.~\ref{fig:flow}. \else The exact construction is detailed in the full paper~\cite{nvw-laprcr-21} and sketched in Fig.~\ref{fig:flow}.\fi

\begin{figure}
	\centering
	\includegraphics[width=1\linewidth]{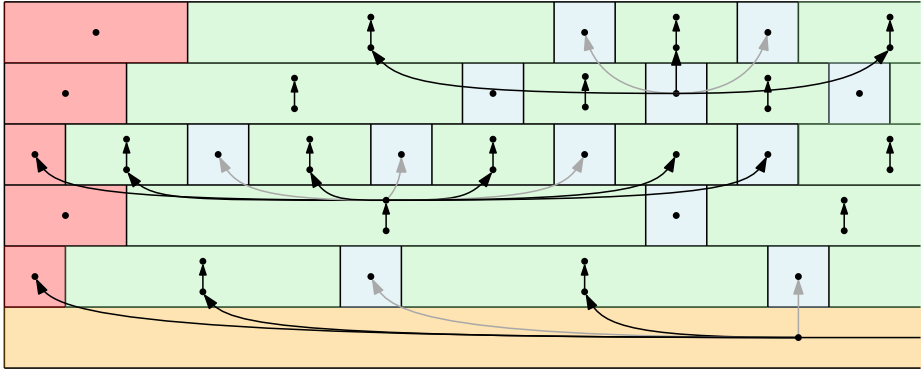}
	\caption{Parts of a flow network: outgoing edges from the source (orange), a rectangle (green) and a gap (blue); red buffer rectangles; gray edges have $\text{cost}$ 1.}
	\label{fig:flow}
\end{figure}

\begin{theorem}\label{thm:flow}
Given a graph~$G=(V,E)$, the cost of a minimum-cost flow~$f$ in $N$ equals the minimum total gap length of any valid representation of $G$. An area-minimal representation of~$G$ is constructed from~$f$ in polynomial time.
\end{theorem}

\begin{proof}
Given any graph $G=(V,E)$, the associated network $N$ has production equal to its consumption
$\sum_{v \in V'} b(v) = b(s) + b(t) = 0$. The source produces $b(s)=w_{max}\cdot K$ flow, which is available to every vertex on layer $1$ (except $v^b$ vertices whose incoming neighbor is $v^a$ on the same layer).
Any vertex $v^a_{1,j}$ must receive $w_{1,j}$ units of flow, as its only edge towards $v^b_{1,j}$ has capacity constraint $c=l=w_{1,j}$. 
Because $w_{1,j}\leq w_{max}$ and since there are at most $K$  vertices $v^a_{1,j}$ on layer $1$, the capacities can be satisfied.
Any edge in $E'$ goes from layer $i$ to $i+1$, except for $(v^a_{i,j},v^b_{i,j})$, but the exact amount of flow that comes from layer $i-1$ into $v^a_{i,j}$ will go through $v^b_{i,j}$ to layer $i+1$. 
Hence for the same reason as in layer $1$, there is enough flow to satisfy the edge capacity constraints, while excess flow is routed through (0-cost) buffers or (1-cost) gaps.

Since in this network only flow that goes into a gap vertex has a (non-zero) cost, flow into gap vertices, and therefore also the total gap width, is minimized. The minimum cost procedure finds this optimal flow $f$.

We construct a minimum-area representation $\mathcal{R}$ by placing rectangles row by row: we leave buffers and gaps equal to the flow routed through the corresponding vertices, and align all rows on the left. The total width of each row, including buffers and gaps, is $w_{max}\cdot K$. Since the flow through each rectangle is constant, the area of the buffers is maximized, to minimize the area occupied by gaps. \qed
\end{proof}

While this method minimizes the area occupied by the drawing it will not always lead to the representation with the minimum bounding box. 
To minimize the size of the bounding box, we propose to limit the amount of flow outgoing from the source node and incoming into the sink node. 
If the chosen bound is too small then the flow network will not be realizable. 
We can thus perform a binary search between the width of the longest layer  $W_{\max}$ as a lower bound and $w_{\max}\cdot K$ as an upper bound. 
Assuming input widths have integer values, this method would add a $\mathcal{O}(\log (w_{\max}\cdot K-W_{\max}))$ factor to the flow runtime.

\section{Maximization of realized contacts}\label{sec:contact-max}

In this section we propose algorithms that maximize the number of realized contacts. We start with a linear-time algorithm for $L = 2$, followed by an integer linear programming model for $L > 2$. The complexity for $L>2$ remains open.

\subsection{Linear-time algorithm for $L=2$}

In this section we describe an algorithm~$\mathcal{A}$ for the case where the input has 2 layers. On a 2-layer graph, a vertex either has one neighbor in the adjacent layer, or more than one. If a vertex has one neighbor we call it a \emph{T-vertex} and if there are more neighbors, it is called a \emph{fan}. A \emph{block} is a maximal sequence of consecutive rectangles in a layer $i$, for which each horizontal contact is realized. A block from the $j$th until the $l$th vertex of row $i$ is the sequence $(R_{i,j},\ldots,R_{i,l})$, where for each $k \in [j, l-1]$ holds that $(R_{i,k}, R_{i,k+1}) \in \mathcal{R}$.

In a given 2-layer graph, there will always be a layer that starts on a fan, while the opposite layer starts on a T-vertex. Assume without loss of generality that $R_{0,0}$ is a fan, otherwise swap the two rows for the duration of the algorithm. 
Algorithm~$\mathcal{A}$ starts by placing $R_{0,0}$, followed by all its neighbors on the adjacent layer, from left to right, ending with $R_{1,j}$. Rectangle~$R_{1,j}$ is again a fan, and the process of placing all opposite-row neighbors, left to right, is repeated for $R_{1,j}$ and every consecutive fan, as they are encountered. We call this placement ordering $\prec$.

When we add a rectangle $R_i$ (fan or T-vertex), we always first attempt to add it next to its horizontal predecessor, if possible (no false adjacency). Though, if the horizontal predecessor is too far left, we place $R_i$ in the leftmost allowed position.
Let $R_0$ be the first rectangle in $\prec$, which is placed on position $x_0$ by $\mathcal{A}$. Algorithm~$\mathcal{A}$ then proceeds by adding $R_1$, representing a T-vertex in the opposite row. Rectangle $R_1$, with width~$w_1$, is placed leftmost, on coordinate $x_0+\varepsilon - w_1$.
We then proceed to add all rectangles corresponding to other T-vertices of $R_0$ one by one, such that all horizontal contacts are realized. Once a T-vertex $R_i$ cannot reach $R_0$, we store the amount of contacts currently realized by $R_0$ as well as its position $x_0$, and slide $R_0$ rightward, to the leftmost position~$x_0'$ that allows a contact of $\varepsilon $ with $R_i$. Note that, since we placed $R_1$ in the leftmost position that allowed a contact of width $\varepsilon$ with $R_0$, we lose at least one contact by moving $R_0$ rightward. If placing $R_0$ at $x_0'$ ties the amount of contacts of $x_0$, then we set $x_0 \coloneqq x_0'$. If $x_0'$ is strictly worse, then the representation is reset to having $R_0$ at $x_0$. From that point on, every time we add a new rectangle, we attempt this shift of the fan and update the position when we find a tie or when we realize more contacts. We repeat this operation for each rectangle, following the order $\prec$, always shifting the last encountered fan.

However, once we consider a fan $R_f$ that is not $R_0$, any sliding operation will be attempted on the block containing $R_f$, rather than just $R_f$. As before, we always shift the block to the leftmost position that realizes the contact between $R_f$ and the newly placed rectangle. We remember position~$x_f$ that realizes most contacts, and favor the newest position on a tie. In case moving the block containing $R_f$ leads to strictly less contacts, we also try to move only $R_f$ instead. This starts a new block containing just $R_f$.
Below we sketch the proof for Theorem~\ref{thm:2layer}, the complete proof can be found in \ifArxiv Appendix~\ref{app:alg}. \else the full paper version~\cite{nvw-laprcr-21}.
\fi

\begin{theorem}\label{thm:2layer}
Algorithm~$\mathcal{A}$ computes a contact maximal valid representation with contacts of length at least $\varepsilon$ for a given 2-layer graph~$G$ in linear time.
\end{theorem}
\begin{proof}[Sketch]
We show that during algorithm~$\mathcal{A}$, the invariant holds that a representation of the first $n$ rectangles in $\prec$ maximizes the number of contacts.

We assume that the invariant holds after $\mathcal{A}$ placed $n-1$ rectangles, such that the current representation $\mathcal{R}^*$ is a contact maximal representation of the first $n-1$ rectangles in $\prec$, and realizes $k$ contacts. Algorithm~$\mathcal{A}$ now adds the next rectangle~$R_n$. We show that the new representation is contact maximal.

\begin{figure}[t]
	\centering    
	\includegraphics{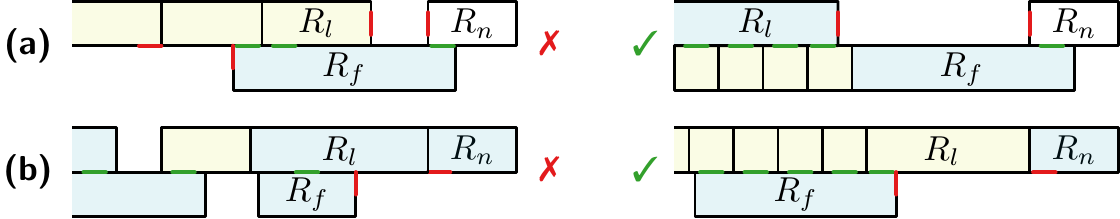}
	\caption{Two configurations where $R_l$ is necessarily a fan (blue) or a T-vertex (yellow). \textbf{(a)} If $R_n$ can achieve only a vertical contact, $R_f$ is a fan. \textbf{(b)} If $R_n$ can achieve only a horizontal contact and is a fan, $R_f$ is a T-vertex.}
	\label{fig:othercases}
\end{figure}

We first prove that the maximum number of contacts that the new representation can realize is $k+2$, since $R_n$ can achieve at most one vertical and one horizontal contact. If these contacts happen naturally, when placing $R_n$ leftmost, then the invariant trivially holds. We therefore prove via a case distinction that in all other cases $k+1$ adjacencies are optimal. We distinguish between the contact that is achieved by $R_n$, either vertical or horizontal. A sole vertical contact with fan (necessarily, forced by the placement ordering) $R_f$ is achieved only if the horizontal predecessor~$R_l$ of $R_n$ is a fan, as shown in Fig.~\ref{fig:othercases}a. If $R_n$ is a T-vertex, a single horizontal contact can arise only if $R_f$ is not moved to $R_n$, as can be seen in Fig.~\ref{fig:casedistinct}. However, when $R_n$ is a fan, this requires $R_l$ to be a T-vertex, see Fig.~\ref{fig:othercases}b. The invariant will therefore still be true after $\mathcal{A}$ added all rectangles, producing a contact maximal representation.

Algorithm~$\mathcal{A}$ considers each rectangle either once, or its degree many times, when a fan is shifted. As the input graph $G$ is planar, $\mathcal{A}$ runs in linear time. \qed
\end{proof}

\begin{figure}[b]
	\centering    
	\includegraphics{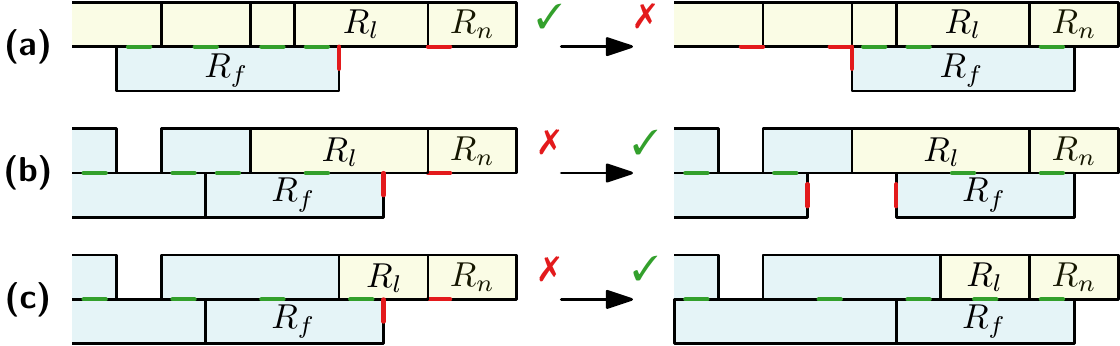}
	\caption{Three configurations where T-vertex $R_n$ does not realize vertical contacts with $R_f$ initially. We move $R_f$ and either \textbf{(a)} reset if the number of contacts is strictly worse, or save when we find \textbf{(b)} a tie, or \textbf{(c)} an increase in contacts.}
	\label{fig:casedistinct}
\end{figure}

\subsection{ILP}

To solve the contact maximization problem on $L > 2$ layers we propose an ILP formulation, which intuitively works as follows. We create a binary contact variable $c(e)$ for each edge~$e$ in the input graph. If a contact is not realized, we set $c(e)=1$ to satisfy the position constraints, otherwise we can set $c(e)=0$.
To handle false adjacencies we add for each rectangle a constraint on the first false contact that happens from the right and left on the row above, if they exist. We use hard constraints on the rectangle coordinates to prevent the false adjacencies. The objective is to minimize the sum over all contact variables, under all these constraints, to maximize the number of realized contacts in a solution. Additional details can be found in \ifArxiv Appendix~\ref{app:ILP}. \else the full paper~\cite{nvw-laprcr-21}.
\fi
\begin{equation}
  \label{eq:sum-formula}
  \text{minimize} \sum_{(v,v')\in E} c(v,v')
\end{equation}

The following inequalities ensure that there is no overlap between rectangles on the same layer (2), and check whether the horizontal contact is realized (3).
\begin{align}
	x_{i,j} + w_{i,j} &\le x_{i,j+1} &&\forall (v_{i,j},v_{i,j+1})\in E\\
	x_{i,j+1} &\le x_{i,j} + w_{i,j} + c(v_{i,j},v_{i,j+1}) M &&\forall (v_{i,j},v_{i,j+1})\in E
\end{align}

The next inequalities verify that the contacts between rectangle $R_{i,j}$ and all of its neighbors on layer $i+1$ are realized.
\begin{align}
	x_{i+1,j'} &\le x_{i,j} + w_{i,j} - \varepsilon + c(v_{i,j},v_{i+1,j'})M &&\forall e(v_{i,j},v_{i+1,j'})\in E\\
	x_{i,j} &\le x_{i+1,j'} + w_{i+1,j'} - \varepsilon + c(v_{i,j},v_{i+1,j'})M &&\forall e(v_{i,j},v_{i+1,j'})\in E
\end{align}

Finally, we model false adjacencies using pairs $(v_{i,j}, v_{i+1, j'})$ stored in sets $F_L$ (resp. $F_R$) that indicate the index of the first block in row $i+1$ that is left (resp. right) of a neighbor of $R_{i,j}$, but is not itself a neighbor of $R_{i,j}$.
\begin{align}
	x_{i+1,j'}+w_{i+1,j'} &\le x_{i,j} &&\forall (v_{i,j},v_{i+1,j'})\in F_L\\
	x_{i,j}+w_{i,j} &\le x_{i+1,j'} &&\forall (v_{i,j},v_{i+1,j'})\in F_R
\end{align}

\bibliographystyle{splncs04}
\bibliography{GD-ref}

\ifArxiv
\newpage
\appendix
\section{Flow construction}\label{app:flow}

For each $v_{i,j} \in V$ we create two copies $v^a_{i,j}$ and $v^b_{i,j}$ in $V'$; for each pair $v_{i,j}, v_{i,j+1}$ we introduce a gap vertex $g_{i,j}$, for each layer a left and right buffer vertex $l_i$ and $r_i$ and a global source and sink, $s$ and $t$. 
All the vertices $v$ have $b(v)=0$ except $b(s)=w_{max}\cdot K$ and $b(t)=-w_{max}\cdot K$ with $w_{max}$ the width of the widest rectangle and $K$ the maximum number of rectangles per layer.

Unless stated otherwise, each edge $e$ has $c(e)=\infty$, $l(e)=0$ and $cost(e)=0$. 
The rectangle-rectangle edges for two vertically adjacent vertices $u$ and $v$ are from the $u^b$ vertex on layer $i$ to the $v^a$ vertex on layer $i+1$. 
For each $v^a_{i,j},v^b_{i,j}$ pair in $G'$ we add an edge $e_{i,j}$ from $v^a_{i,j}$ to $v^b_{i,j}$ with $c(e_{i,j})=w_{i,j}$ and $l(e_{i,j})=w_{i,j}$.
This ensures that a vertex must receive exactly as much flow as the width of the rectangle it represents. 
We also add edges going in and out of the buffer vertices: to minimize the gaps it is preferable for the outer rectangles to get missing flow from the buffers on the outside or to route excess flow to these buffers.

Edges from rectangle vertices to gap vertices are defined as follows.
For a vertex $v^b_{i,j}$, we add an edge to all vertices $v^a$ representing its neighbors in $G$ on layer $i+1$, and we add edges to the gaps left and right of those neighbors. 
All edges to gaps have $cost=1$, to penalize the creation of gaps.
Lastly we add edges from gaps to other vertices. 
For a gap vertex $g_{i,j}$ we look at its left and right neighbors $v^b_{i,j}$ and $v^b_{i,j+1}$ respectively. We add edges from $g_{i,j}$ to all rectangles and gaps that have an incoming edge from $v^b_{i,j}$ and $v^b_{i,j+1}$, and assign $cost=0$ for edges into rectangles, and $cost=1$ for edges into gaps. We do not penalize the outgoing flow of a gap, since we already count the incoming flow.

From this network we can easily deduce that the minimum cost flow corresponds to the solution that minimizes total gap length: only gap vertices have a cost, hence the flow avoids these vertices whenever possible. We can construct the network and then solve the minimum flow problem in polynomial time. In the obtained solution, the flow values found on each edge should represent the length of the overlap between these elements, which allow us to construct the corresponding representation. However, in practice the relationship between units of flow and contact length is not always direct. Indeed, consider a 4 vertex configuration where $a$ lies bottom left, $b$ bottom right, $c$ top left and $d$ top right, and both $a$ and $b$ have an edge to both $c$ and $d$. Note that in such a case $a$ and $d$, or $c$ and $b$, would be gaps. It is possible that $a$ favors sending its flow to $d$ and $b$ to $c$, which is impossible to represent as contacts in a configuration of rectangles.

Once the flow has been computed, we can locally swap the required flow between the crossing edges to resolve those \emph{crossing patterns}.

\section{Linear-time algorithm -- Omitted proof}\label{app:alg}

\begin{usecounterof}{theorem}{thm:2layer}
	Algorithm~$\mathcal{A}$ computes a contact maximal valid representation with contacts of length at least $\varepsilon$ for a given 2-layer graph~$G$ in linear time.
\end{usecounterof}
\begin{proof}
	We show that during the placement of rectangles by algorithm~$\mathcal{A}$, the invariant holds that a representation computed for the first $n$ rectangles in $\prec$ achieves a maximum number of contacts.
	First observe that, we start by placing $R_0$ and $R_1$, such that their only contact is realized. The invariant therefore holds at the start. 
	
	Now assume that the invariant holds after $\mathcal{A}$ placed $n-1$ rectangles, such that the current representation $\mathcal{R}^*$ is a contact maximal representation of the first $n-1$ rectangles in $\prec$, and realizes $k$ contacts. Algorithm~$\mathcal{A}$ now adds the next rectangle~$R_n$.
	We show that the new representation is contact maximal.
	
	The maximum number of contacts that the new representation can realize is $k+2$,
	because rectangle $R_n$ can achieve at most one vertical and one horizontal contact. If there was a representation of the $n$ rectangles that realized $k'>k+2$ contacts, then removing $R_n$ would leave $k'-2>k$ contacts, which contradicts our assumption that $\mathcal{R}^*$ is optimal.
	If $R_n$ naturally realizes its vertical and horizontal contacts, the representation realizes $k+2$ contacts.
	There is no configuration where we cannot realize any contact when adding rectangle~$R_n$ to $\mathcal{R}^*$, since one of the rows extends further, and hence $R_n$ either can achieve its horizontal contact, or the opposing row extends at least $\varepsilon$ and the vertical contact can be realized. We therefore consider only the cases where adding $R_n$ results in a representation with $k+1$ contacts.
	Throughout the rest of the proof, we refer to the vertical fan neighbor of $R_n$, that has already been placed, as $R_f$. $R_f$ necessarily lies in the opposite layer. The only rectangles that will be moved when placing $R_n$ are the rectangles in the block containing $R_f$. Any other rectangles and adjacencies are untouched. So it is sufficient to count the adjcencies gained and lost by the block containing $R_f$ and $R_n$ to confirm that the representation of $n$ rectangles is optimal.
	
	\begin{itemize}
		\item If $R_n$ realizes only its vertical contact with $R_f$, then rectangle~$R_l$, neighboring $R_n$ on the left, is necessarily a fan (see Fig.~\ref{fig:othercases2}a)). Assume for contradiction that $R_l$ is a T-vertex. This implies that $R_f$ is $R_l$'s only allowed vertical contact. Since $R_n$ achieves contact with $R_f$ and $R_f$ is leftmost, $R_l$ must also be in contact with $R_f$. As a result the horizontal contact with $R_l$ always happens, contradicting our assumption that $R_n$ has only a vertical contact. As a consequence, $R_l$ is a fan, and the position for a fan is the position that maximizes contacts, favoring newer positions (or positions more to the right) for ties. An alternative to $\mathcal{R}^*$ with $R_l$ in contact with $R_f$ must therefore achieve less than $k$ contacts.
		Additionally, $R_n$ cannot be slid left of $R_f$ even though there is a gap towards $R_l$, as the rectangles left of $R_f$ are not part in the neighborhood of $R_n$. Thus, sliding $R_n$ further left would cause false adjacencies. There is no contact maximal representation that preserves the contact between $R_l$ and $R_n$, and hence $k+1$ contacts is maximal.
		
		\item If $R_n$ realizes only its horizontal contact, then it fails the contact with $R_f$. We distinguish between $R_n$ being a fan or a T-vertex.
		\begin{itemize}
			\item If $R_n$ is a T-vertex, then $R_f$ is too far to the left, and the algorithm will try to create representations where $R_f$ is in the leftmost position that realizes a contact with $R_n$: either by moving the block containing $R_f$, or alternatively moving just $R_f$.
			\begin{figure}[t]
				\centering    
				\includegraphics{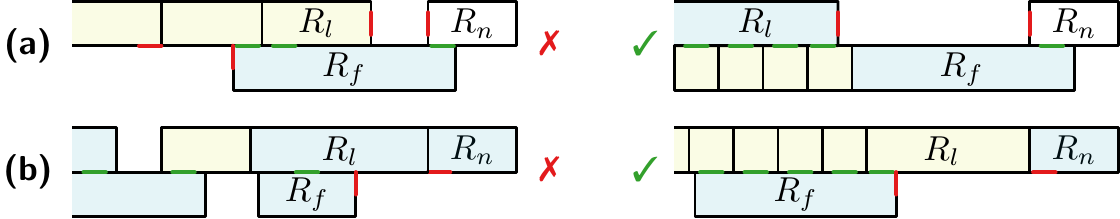}
				\caption{Two configurations where $R_l$ is necessarily a fan (blue) or a T-vertex (yellow). \textbf{(a)} If $R_n$ can only achieve a vertical contact, $R_f$ is a fan. \textbf{(b)} If $R_n$ can only achieve a horizontal contact and is a fan, $R_f$ is a T-vertex.}
				\label{fig:othercases2}
			\end{figure}
			\begin{itemize}
				\item If $R_f$ remains in its original position, which does not achieve the vertical contact with $R_n$, it means that the leftmost position of $R_f$, that achieves the contact with $R_n$ realizes strictly less contacts than the original position of $R_f$ in $\mathcal{R}^*$. As a result, there is no contact maximal representation that preserves the contact between $R_f$ and $R_n$, and thus $k+1$ contacts is optimal (see Fig.~\ref{fig:casedistinct2}a).
				\item When the block containing $R_f$ is shifted to the right to create the alternative representation that ties the number of contacts of $\mathcal{R}^*$, at most one vertical contact must have been lost, to gain the vertical one. Otherwise, $R_f$ can give up its horizontal contact to achieve the same result (see Fig~\ref{fig:casedistinct2}b). Thus, while $R_n$ realises 2 contacts, $R_f$ lost one and the representation realizes $k+1$ contact. This alternative representation is optimal and used by $\mathcal{A}$ instead of the representation where $R_f$ is not moved rightward.
				\item If the block containing $R_f$ does not lose contacts when shifted, then the overall representation will gain the contact between $R_n$ and $R_f$ contact. This configuration is optimal because it achieves $k+2$ contacts (see Fig.~\ref{fig:casedistinct2}c).
			\end{itemize}
			\item If $R_n$ is a fan, then rectangle $R_l$, again neighboring $R_n$ on the left, must be a T-vertex, since $R_f$ cannot be an empty fan (see Fig.~\ref{fig:othercases2}b). Assume for contradiction that $R_l$ is a fan, then $\mathcal{A}$ would try to realize the vertical contacts between $R_f$ and the fans $R_l$ and $R_n$ on the other row. In this case $R_f$ would necessarily be an empty fan, and hence $\mathcal{A}$ will try to sacrifice the contact with the left neighbor of $R_f$, to gain the vertical contact between $R_f$ and $R_n$. As there are no T-vertices for $R_f$ to lose contacts with, $\mathcal{A}$ will always find at least a tie between the configuration that has the vertical contact between $R_f$ and $R_n$, and $\mathcal{R}^*$, and hence prefer the new configuration. This contradicts our assumption that $R_n$ realizes only its horizontal contact, and hence $R_l$ must be a T-vertex. When $R_n$ is added, an alternative representation is created with $R_f$ shifted to the leftmost position that realizes the contact with $R_n$. Since this configuration is not chosen by~$\mathcal{A}$, the configuration realizes less than $k+1$ contacts, and therefore $k+1$ contacts is optimal.
		\end{itemize}
	\end{itemize}
	Thus, after adding $R_n$ to $\mathcal{R}^* $, algorithm~$\mathcal{A}$ produces a representation that realizes $k+2$ contacts, if it is possible, and $k+1$ otherwise, which is optimal. The invariant will therefore still be true after $\mathcal{A}$ added all rectangles, producing a contact maximal representation of our input graph.
	
	The algorithm considers each rectangle either only once, or as many times as its degree when an alternative shifted representation is created. As the input graph is planar, the algorithm runs in linear time. \qed
\end{proof}

\begin{figure}[t]
	\centering    
	\includegraphics{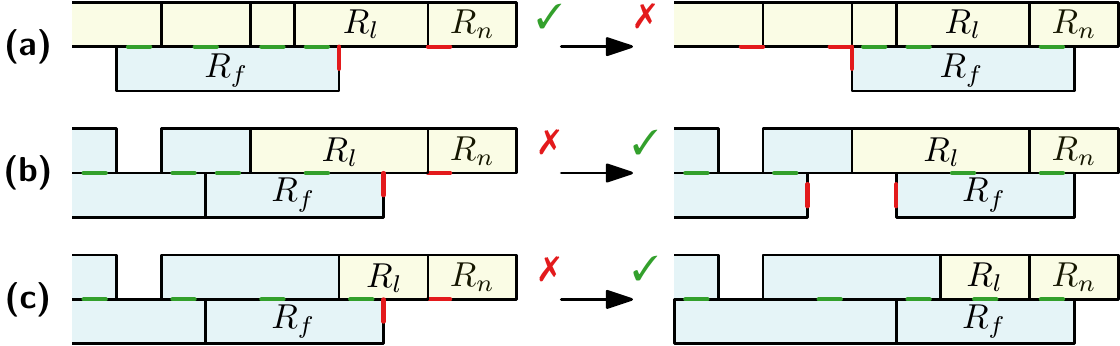}
	\caption{Three configurations where T-vertex $R_n$ does not realize vertical contacts with $R_f$ initially. We move $R_f$ and either \textbf{(a)} reset if the number of contacts is strictly worse, or save when we find \textbf{(b)} a tie, or \textbf{(c)} an increase in contacts.}
	\label{fig:casedistinct2}
\end{figure}

As one may already have realized, in some cases, when looking at two sequential fans that are not in contact, two T-vertices will be in a point contact. This happens when the T-vertices cannot overlap without creating a false adjacency, as we show in see Fig~\ref{fig:pointcontact}.

\begin{figure}[b]
	\centering    
	\includegraphics{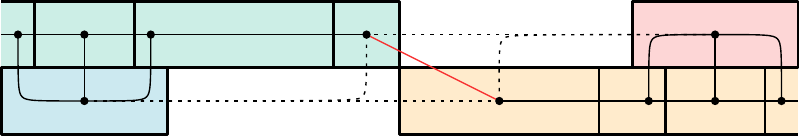}
	\caption{Edge case where the prefered position of the fan vertices (red and blue) causes a point contact on fan vertices (orange and green)}
	\label{fig:pointcontact}
\end{figure}

\section{ILP model}\label{app:ILP}

In this section we present an integer linear programming formulation $M$, which can be used to find optimal solutions to the contact maximization problem. In this formulation, the constants are the width~$w_{i,j}$ of each rectangle, the list~$E$ of contacts, the minimal length~$\varepsilon$ of a contact and a large integer~$M$ for indicator type constraints. Additionally, we have the lists~$F_L$ and $F_R$ of false adjacencies, defined as follows. Consider the rectangle representing vertex $v_{i,j}$ and its neighbors $N = \{v_{i+1,k},\ldots,v_{i+1,l}\}$ on the row above. For each such vertex $v_{i,j}$, we store the pairs $(v_{i,j}, v_{i+1,k-1})\in F_L$ and $(v_{i,j},v_{i+1,l+1})\in F_R$, if $v_{i+1,k-1}$ and $v_{i+1,l+1}$ exist. These pairs represent the false adjacencies of $R_{i,j}$ with the rightmost rectangle left of $N$ and the leftmost rectangle right of $N$, respectively. Note that preventing these two false adjacencies will prevent all false adjacencies of $R_{i,j}$ with rectangles in row $i+1$, as long as the order of the rectangles in both rows is correctly maintained.

The variables of the ILP are the following. We use $x_{i,j}$ as variables for the position of each rectangle, representing the bottom left corner of the rectangle, and Boolean variables~$c(v,v')$ which indicate that a contact is not realized between $v$ and $v'$.
We start by stating the optimization function, which minimizes the sum of all $c$ variables. In practice, this means we maximize realized adjacencies.

\begin{equation}
  \label{eq:sum-formula}
  \text{minimize} \sum_{(v,v')\in E} c(v,v')
\end{equation}

The following inequalities ensure firstly, that there is no overlap between rectangles on the same layer, and secondly, check whether the horizontal contact is realized: if $x_{i,j}$ is too small, then $c$ must be set to $1$, and hence the represented rectangle is too far left to have the contact naturally.
\begin{align}
	x_{i,j} + w_{i,j} &\le x_{i,j+1} &&\forall (v_{i,j},v_{i,j+1})\in E\\
	x_{i,j+1} &\le x_{i,j} + w_{i,j} + c(v_{i,j},v_{i,j+1}) M &&\forall (v_{i,j},v_{i,j+1})\in E
\end{align}

The next inequalities ensure that the contacts between rectangle $R_{i,j}$ and all of its neighbors on layer $i+1$ are realized, such that the left side of a neighbor is left of the right side of $R_{i,j}$, or symmetrically, the right side of a neighbor is right of the left side of $R_{i,j}$. Again, if we cannot satisfy the inequality with $c=0$, then it is set to 1, and since $M$ is a large value, the inequalities are trivially satisfied.
\begin{align}
	x_{i+1,j'} &\le x_{i,j} + w_{i,j} - \varepsilon + c(v_{i,j},v_{i+1,j'})M &&\forall e(v_{i,j},v_{i+1,j'})\in E\\
	x_{i,j} &\le x_{i+1,j'} + w_{i+1,j'} - \varepsilon + c(v_{i,j},v_{i+1,j'})M &&\forall e(v_{i,j},v_{i+1,j'})\in E
\end{align}

Finally, we model false adjacencies using the pairs stored in $F_L$ and $F_R$. For a pair in $F_L$, the rectangle $R_{i+1,j'}$, with which $R_{i,j}$ has a false adjacency, should stay left of $R_{i,j}$. This forces the left side of $R_{i,j}$ to be right of the right side of $R_{i+1,j'}$. Symmetrically, pairs in $F_R$ prevent false adjacencies when $R_{i+1,j'}$ is to the right of $R_{i,j}$.
\begin{align}
	x_{i+1,j'}+w_{i+1,j'} &\le x_{i,j} &&\forall (v_{i,j},v_{i+1,j'})\in F_L\\
	x_{i,j}+w_{i,j} &\le x_{i+1,j'} &&\forall (v_{i,j},v_{i+1,j'})\in F_R
\end{align}

\begin{theorem}
	Solving ILP model~$M$ optimally, results in an optimal solution for the contact maximization problem.
\end{theorem}
\begin{proof}
First, we show that whenever a variable~$c$ is set to zero, then a contact is realized. Let us assume that $c = 0$, then one of the following two cases applies.
\begin{itemize}
	\item If $c(v_{i,j},v_{i,j+1})=0$, then the horizontal contact between $R_{i,j}$ and $R_{i,j+1}$ is realized. Since $x_{i+1,j} \le x_{i,j} + w_{i,j}$ and $x_{i,j} + w_{i,j} \le x_{i,j+1}$, then $x_{i,j} + w_{i,j} = x_{i,j+1}$. Thus the coordinate of the right side of $R_{i,j}$ lies on the left side of $R_{i,j+1}$, and there is a horizontal contact.
	\item If $c(v_{i,j},v_{i+1,j'})=0$, then the vertical contact between $R_{i,j}$ and $R_{i+1,j'}$ is realized. We have both $x_{i+1,j'} \le x_{i,j} + w_{i,j} - \varepsilon$ which means that the left side of $R_{i,j}$ lies left of the right side of $R_{i+1,j'}$ and $x_{i,j} \le x_{i+1,j'} + w_{i+1,j'}$ which means that the right side of $R_{i,j}$ lies right of the left side of $R_{i+1,j'}$.
\end{itemize}

Finally, since we minimize the sum of the $c$ values, as many as possible are set to $0$. Each $c$ value set to zero corresponds to a realized adjacency, and thus the number or contacts in the resulting representation is maximized. \qed
\end{proof}
\fi
\end{document}